\ifdefined\inmainpaper
\else
\documentclass[11pt,a4paper]{article}

\usepackage[margin=3cm]{geometry}

\usepackage{amsmath,amssymb,amsthm}
\usepackage{amssymb}
\usepackage{hyperref}
\usepackage{graphicx}
\usepackage{epsfig}
\usepackage{xspace}
\usepackage{optidef}
\usepackage{algorithm}
\usepackage{algorithmic}
\usepackage{booktabs}
\usepackage{amsthm} 

\newcommand{\negA}{\vspace{-0.0in}}  
\newcommand{\negB}{\vspace{-0.0in}}  
\newcommand{\negC}{\vspace{-0.0in}}  

\newcounter{algsubstate}


\newtheorem*{claimstar}{Claim}
\title{Bipartite Matching with Pair-Dependent Bounds}
\author{
  Shaul Rosner\thanks{Tel Aviv University. \texttt{srosner@tauex.tau.ac.il}. 
    ORCID: \href{https://orcid.org/0009-0006-4671-2566}{0009-0006-4671-2566}}
  \and
  Tami Tamir\thanks{Reichman University. \texttt{tami@runi.ac.il}. 
    ORCID: \href{https://orcid.org/0000-0002-8409-562X}{0000-0002-8409-562X}}
}
\date{}
\usepackage{amsfonts,amssymb,mathtools}
\newtheorem{theorem}{Theorem}
\newtheorem{corollary}[theorem]{Corollary}
\newtheorem{claim}[theorem]{Claim}
\newtheorem{definition}[theorem]{Definition}
\newtheorem{observation}[theorem]{Observation}

\newcommand{\set}[1]{\{ #1 \}}

\newcommand\I{{\cal{I}}}

\begin{document}
\maketitle
\begin{abstract}


Let $G=(U \cup V, E)$ be a bipartite graph, where $U$ represents jobs and $V$ represents machines. We study a new variant of the bipartite matching problem in which each job in $U$ can be matched to at most one machine in $V$, and the number of jobs that can be assigned to a machine depends on the specific jobs matched to it.
These pair-dependent bounds reflect systems where different jobs have varying tolerance for congestion, determined by the specific machine they are assigned to.

We define a {\em bipartite PD-matching} as a set of edges $M \subseteq E$ that satisfies these job-to-machine tolerance constraints. This variant of matching extends well-known matching problems, however, despite its relevance to real-world systems, it has not been studied before. We study bipartite PD-matchings with the objective of maximizing the matching size. As we show, the problem exhibits significant differences from previously studied matching problems. We analyze its computational complexity both in the general case and for specific restricted instances, presenting hardness results alongside optimal and approximation algorithms.

\end{abstract}
\thispagestyle{empty}
\newpage
\fi

\setcounter{page}{1}
\section{Introduction}
Finding a maximum matching in a bipartite graph is one of the most fundamental problems in combinatorics and graph theory. 
Formally, a bipartite graph is a graph $G=(U \cup V, E)$ in which
$E \subseteq U \times V$. A matching in $G$ is a set of edges, $M \subseteq E$, such that each vertex in $U \cup V$ is an endpoint of at most one edge in $M$. In other words, each vertex in $U$ is matched with at most one vertex in $V$ and vice-versa.
A maximum matching is a matching with maximal cardinality. 

Bipartite matching has diverse practical uses (\cite{FlowsBook},
Section 12.2). Most notably, bipartite graphs are used to model assignments of jobs to machines. Specifically, the $U$-vertices correspond to jobs, the $V$-vertices correspond to machines, and there is an edge $(j,i)$ if and only if machine $i$ is capable of processing job $j$. 



Our work is motivated by scheduling environments in which all jobs assigned to a machine are processed in parallel using time-multiplexing. 
In such an environment, jobs are associated with {\em machine-dependent tolerances}. That is, for every job $j$ and machine $i$, we are given the {\em congestion tolerance} of $j$ if processed by machine $i$.

Consider for example a large-scale online education platform that schedules live mentoring sessions between students (jobs) and instructors (machines). Each instructor can teach several students simultaneously, however, the quality of interaction and learning decreases as more students are assigned to the same instructor.
To preserve learning effectiveness, each student has a tolerance: a maximum number of students they are willing to share the session with, depending on the instructor. For instance, an experienced instructor may be more engaging and thus tolerable even in a larger group, while students may prefer more individualized attention from less experienced mentors.

As another example, consider a cloud computing environment in which virtual machines (VMs) are deployed on physical servers that use time-sharing to run multiple VMs in parallel. The performance experienced by each VM depends not only on the number of VMs sharing the server but also on the specific server’s capabilities. For example, high-performance servers can support more concurrent VMs without significant degradation, while less powerful servers offer acceptable performance only under lighter loads. Each VM has a server-dependent tolerance—representing the maximum number of co-located VMs it can share a server with while meeting its performance needs. The goal is to assign VMs to servers in a way that respects these tolerances and maximizes the number of successfully deployed VMs.

%
%
These examples illustrate that machine-related tolerances naturally arise in real-life applications, motivating our study beyond its theoretical interest.

We model the above job-to-machine assignment problem as a new variant of bipartite matching:{\em bipartite matching with pair-dependent bounds} (PD-matching). In this model, every $U$-vertex can be matched at most once, and the number of times a $V$-vertex can be matched depends on which vertices it is matched to. 
Recall that if $k$ unit-length jobs are assigned to a machine that employs time-sharing processing, they all experience congestion $k$. Consequently, if a job $j \in U$ has tolerance $k$ for machine $i \in V$, then in every PD-matching that includes the edge $(j,i)$, the degree of $i$ cannot exceed $k$. 

The maximum PD-matching problem is therefore equivalent to the fundamental problem of maximizing the throughput, given by the number of satisfied jobs, in a scheduling environment with time-sharing processing and machine-dependent tolerances. Surprisingly, to the best of our knowledge, this problem
has not been previously studied. As we show, the problem exhibits significant differences from previously studied scheduling problems, as well as known variants of bipartite matching.

We identify specific classes of instances where a maximum PD-matching can be computed in polynomial time, as well as classes where the problem is computationally hard. An instance is characterized by the $|U| \times |V|$ tolerance matrix. We examine various properties of this matrix, such as monotonicity, limited number of values, or limited row-types, and analyze the computational complexity of PD-matching for instances in the corresponding class. 
Beyond its theoretical significance, our analysis provides insights into the effects of time-sharing processing and machine-dependent tolerances on the systems throughput.

\section{Problem Statement and Preliminaries}
\label{sec:model}

Let $G=(U\cup V, E)$ be a simple complete bipartite graph with edge set $E = U \times V$. 
Motivated by a job-to-machine assignment problem, we refer to the vertices of $U$ as jobs, and to the vertices of $V$ as machines. An instance of bipartite matching with {\em pair-dependent} bounds (PD-matching, for short) is given by $I=\langle n,m, \set{b(j,i)}_{1 \le j \le n, 1 \le i\le m}  \rangle$, where $n=|U|$ and $m = |V|$, and for every $(u,v) \in U \times V$, $b(j,i) \ge 0$ is the {\em tolerance} of job $j$ on machine $i$. Note that $G$ is a complete bipartite graph, however, by setting $b(j,i)=0$, it is possible to prevent an assignment of $j \in U$ to $i \in V$.

A {\em PD-matching} is a subset of edges $M \subseteq E$ that fulfill the pair-dependent bounds. Formally, for every vertex $i \in V$, let $\Gamma_i(M)$ be the set of $U$-vertices matched to vertex $i$ in $M$, that is, $j \in \Gamma_i(M)$ iff $(j,i) \in M$. The degree of vertex $i$ in $M$, denoted by $d_i(M)$, is the number of $U$-vertices matched to $i$, that is, $d_i(M)=|\Gamma_i(M)|$. In a PD-matching $M \subseteq E$, each vertex $j \in U$ is incident with at most one edge in $M$, and additionally, for every $(j,i) \in M$, $d_i(M) \leq b(j,i)$. 
We consider the problem of maximizing $|M|$, the cardinality of the PD-matching. In other words, we attempt to find a set of edges $M$ with a maximum number of matched nodes from $U$, without exceeding their tolerance constraints.


Some of our results refer to restricted classes of instances.
An instance is {\em monotonous} if each row and each column in the tolerance matrix is non-decreasing. That is, the $U$-vertices as well as the $V$-vertices can be ordered such that 
for every $1 \leq j < j' \leq n$ and $1 \leq i < i' \leq m$, it holds that $b(j,i) \leq b(j, i')$ and $b(j,i) \leq b(j',i)$. Practically, monotonicity implies that the jobs as well as the machines can be ordered according to their tolerance and capability. We denote by $\I_{mono}$ the class of monotonous instances. 

Other classes of instances, which are motivated by real-life applications, include instances with $U$-dependent or $V$-dependent tolerances. 
%
The class $\I^{Udep}$ includes instances in which every job $j \in U$ is associated with a single parameter $b_j$ and a set $V_j \subseteq V$ of machines that are capable of processing it. 
Formally, an instance is in $\I^{Udep}$ if, for all $j \in U$ and $i \in V$, $b(j,i) \in \{0, b_j\}$ where $b(j,i) =b_j$ iff $i \in V_j$. Practically, this class corresponds to environments with restricted assignment and machine-{\em in}dependent tolerances.
Similarly, the class $\I^{Vdep}$ includes instances in which every machine $i \in V$ is associated with a single parameter $b_i$ and a set $U_j \subseteq U$ of jobs it can process. For all $i \in V$ and $j \in U$, $b(j,i) \in \{0, b_i\}$ where $b(j,i) =b_i$ iff $i \in U_j$. 
Note that the classes $\I^{Udep}$ and $\I^{Vdep}$ are not equivalent, since in a PD-matching, every $u \in U$ is matched to at most one $V$-vertex, while a node $v \in V$ may be matched to multiple $U$-vertices.


\subsection {Our Results}
In this paper, we introduce and study the problem of computing a maximum PD-matching in a bipartite graph. This problem can be interpreted as maximizing a system's throughput in the presence of job-machines dependent tolerances. We analyze its computational complexity both in the general case and for specific restricted instances, presenting hardness results alongside optimal and approximation algorithms.

In Section~\ref{sec:app}, we analyze the performance of natural greedy approaches. A PD-matching $M$ is {\em maximal} if, for any edge not in $M$, we have that $M \cup \{e\}$ is not a PD-matching. Unlike classical bipartite matching problems where every maximal matching is a $\frac 1 2$-approximation to a maximum one, we show that with vertex-dependent bounds, maximal matchings do not provide a constant approximation. On the other hand, we define {\em strongly-maximal} PD-matching, that can be computed by simple greedy algorithms and provide a $\frac 1 2$-approximation. This discussion clarifies the difficulties inherent in problem of a maximum PD-matching.

Recall that an instance of PD-matching is {\em monotonous} if the $U$-vertices and the $V$-vertices can be ordered, respectively, according to their tolerance and capability. 
In Section~\ref{sec:hardM} we prove that computing a maximum PD-matching is NP-hard even for monotonous instances. We believe that this result is counterintuitive and unexpected, since monotonicity is a very strong assumption that is often associated with the existence of a straightforward optimal solution derived through exchange arguments.

In Section \ref{sec:dep} we analyze the classes $\I^{Udep}$ and $\I^{Vdep}$ of instances with $U$-dependent and $V$-dependent bounds.
We first show that computing a maximum PD-matching of an instance in $\I^{Vdep}$ can be reduced to a $b$-matching problem, and therefore, can be done efficiently.
For classes of instances with $U$-dependent bounds, we present two optimal algorithms, and a hardness result. The first optimal algorithm is for instances in $\I^{Udep} \cap \I_{mono}$, and the second is for instances in which for every $u \in U$, we have $V_u=V$. 
Note that the special case of $\I^{Udep}$ where for all $j$ we have $b_j=1$, is the classical bipartite maximum matching problem, which is solvable efficiently~\cite{HK73, CK24}.
We prove that if the $U$-dependent job tolerances are allowed to be in $\{1,2\}$, then the problem becomes APX-hard.

In Section~\ref{sec:restricted} we study additional restricted classes of instances. $(i)$ We provide a polynomial-time algorithm for instances with a constant number of machines. $(ii)$ We study instances with a limited set, $\mathcal{T}$, of job tolerances. We prove the problem is polynomially solvable if $|\mathcal{T}| = 1, \mathcal{T} = \{0, k\}$, or $\mathcal{T}=\{1,2\}$, and is APX-hard otherwise. For monotonous instances, a maximum PD-matching can be computed in polynomial time if $|\mathcal{T}| \leq 3$. $(iii)$ we provide a polynomial-time algorithm for instances with a constant number of job types, where two jobs $j_1, j_2$ have the same type iff for every machine $i \in V$, $b(j_1,i)=b(j_2,i)$.

We conclude in Section \ref{sec:conclusions}, with a discussion and some directions for future work.


\subsection{Related Work}
Matching is one of the most fundamental concepts in combinatorics and graph theory. The basic maximum bipartite-matching problem can be solved efficiently (e.g., in time $O(\sqrt n m)$ by Hopcroft-Karp algorithm~\cite{HK73}, $O(n^{2+o(1)})$ by Chuzhoy and Khanna~\cite{CK24}, and $O(m^{1+o(1)}$ as recently shown by Chen et. al~\cite{CK+23}), as well as numerous well-studied variants~\cite{Lov09}. In particular, in a $b$-matching~\cite{Ans87,KS95}, a function $b:V \rightarrow \mathbf{N}$ is given, and a matching is valid if every vertex has degree at most $b(v)$. 
In a semi-matching~\cite{semi,FLN14}, all vertices in $U$ have degree $1$ and the goal is to balance the load on the $V$-vertices. 
Other variants of bipartite matching, in which vertex degrees are limited to a specific set of values, or when the matching edges should satisfy certain properties are known to be NP-hard~\cite{DP17,TIR78}.
Conflict-aware bipartite $b$-matching, in which certain vertices may not be matched to the same vertex, is analyzed and shown to be NP-hard in~\cite{Chen16}. 

Recall that maximum PD-matching corresponds to the problem of maximizing the throughput in a scheduling problem of assigning unit-length jobs that have machine-dependent tolerances. The tolerance can be viewed as the job's allowed due-date on the machine assuming time-sharing processing.
Scheduling theory is a very well studied field, dating back to the early 1950's~\cite{GLLR79}. 
When jobs are processed sequentially by the machines, the maximum  throughput problem is well-studied. For a single machine, Moore-Hodgson's algorithm (\cite{Moo68}) solves the problem optimally. For multiple machines, the problem is NP-hard even with preemptions allowed~\cite{Law83}.
When jobs have equal-lengths, an optimal solution can be produced using max-flow techniques~\cite{BCS74,Horn73}.
A variant of sequential scheduling with machine-dependent due-dates is discussed in \cite{Hall86,MMS25}. In their setting, each machine has a sorted list of due-dates, and the $i$'th job assigned to a machine has the $i$'th due-date.


The algorithms community has devoted a lot of attention to variants of packing and assignment problems. Maximum PD-matching is a special case of the most general variant, denoted SAP (separable assignment problem)~\cite{FGMS11}. However, it is inherently different from other variants such as GAP (general assignment problem)~\cite{ST93,CK05} and DCP (distributed caching problem)\cite{FGMS11, ST03}.
A SAP problem is defined by a set of bins and a set of items to pack in each bin, a value, $f_{ji}$, for assigning item $j$ to bin $i$; and a separate packing constraint for each bin, i.e., for bin $i$, a family ${\cal U}_i$ of subsets of items that fit in bin $i$. The goal is to pack items into bins to maximize the aggregate value.
Given a $\beta$-approximation algorithm for the single-bin problem, \cite{FGMS11} presents a polynomial-time LP-rounding based $(1-\frac 1 e)\beta$-approximation algorithm for SAP.

An instance of maximum PD-matching can be viewed as an instance of SAP in which every $U$-vertex is an item, and every $V$-vertex is a bin. For every bin $i$ and set $W \subseteq U$, the set $W$ can fit in bin $i$, that is, $W \in {\cal U}_i$, if and only if $|W| \le min_{j \in W} b(j,i)$. 
The single-bin problem is solved optimally, by assigning to machine $i$ the maximal number $k$ of jobs for which $b(j,i) \ge k$. This reduction implies that the approximation algorithm of~\cite{FGMS11} can be applied with $\beta=1$ to yield a $(1-\frac 1 e)$-approximation to maximum PD-matching.

\section{Analysis of Natural Greedy Techniques}
\label{sec:app}
In the most natural greedy approach, edges are added to the matching as long as possible. Formally, a PD-matching $M$ is {\em maximal}, if, for any $e \in E \setminus M$, it holds that $M \cup \{e\}$ is not a PD-matching.
A well-known result regarding the classical bipartite matching problem is that every maximal matching is a $\frac 1 2$-approximation to the maximum matching problem~\cite{Lov09}. 

We show that with arbitrary pair-dependent constraints, maximal PD-matchings may have a poor quality. Specifically, given $r>1$, consider an instance $I_r$ with $|U|=r+1$, $|V|=1$, and the following bounds:  
For the first job, let $b(1,1)=1$, while for each of the other $r$ jobs, let $b(j,1)=r$. It is easy to verify that the edge $(1,1)$ is a maximal PD-matching, while the optimal solution includes the $r$ edges $\{(j,1)~|~2 \le j \le r+1\}$. Thus, a maximal PD-matching may be smaller than a maximum one by factor $r$.

In a maximal PD-matching, for any $e=(j,i) \in E \setminus M$, it holds that $M \cup \{e\}$ is not a valid PD-matching. This invalidity may be either since $M \cup \{e\}$ is not feasible for $j$, specifically, $j$ is already matched or $d_i(M) +1 > b(j,i)$, or since $M \cup \{e\}$ is feasible for $j$, however, the addition of $(j,i)$ will violate the tolerance-constraint of some $j' \neq j$ assigned to $i$. We refine the definition of maximal PD-matching, distinguishing between the above cases.

\begin{definition}
A PD-matching $M$ is {\em strongly-maximal}, if, for any unmatched job $j\in U$ and any machine $i \in V$, it holds that $b(j,i) \le d_i(M)$.    
\end{definition} 
Clearly, any strongly-maximal PD-matching is maximal, however, as demonstrated in $I_r$ above, a maximal PD-matching may not be strongly-maximal. 
We show that every strongly-maximal PD-matching is a $\frac 1 2$-approximation, and then present a simple greedy algorithm for computing one.

\begin{claim}
\label{cl:2approx}
Every strongly-maximal PD-matching $M$ is a $\frac 1 2$-approximation.
\end{claim}
\begin{proof}
Let $M$ be a PD-matching of size $k$ and assume towards contradiction that for a maximum cardinality PD-matching $M^*$, it holds that $|M^*|\ge 2k+1$. This implies that there is a set $W$ of at least $k+1$ jobs that are matched in $M^*$ but not in $M$.

In $M^*$, the jobs in $W$ are matched to a set of machines $V_W$ with total degree at least $|W|=k+1$. In $M$, the total degree of the machines in $V_W$ is at most $k$. Thus, for at least one machine ${i_1} \in V_W$, $d_{i_1}(M)<d_{i_1}(M^*)$. Let $j\in W$ be a job assigned to $i_1$ in $M^*$. Such a job exists since $i_1 \in V_W$. It holds that $d_{i_1}(M^*) \leq b(j, i_1)$. Therefore, $d_{i_1}(M)+1 \leq b(j, i_1)$, and $M$ is not strongly-maximal.
\end{proof}

Next, we provide an efficient algorithm for computing a strongly-maximal PD-matching.

\begin{algorithm}[ht]
\caption{An algorithm for strongly-maximal PD-matching}
\begin{algorithmic}[1]
\label{alg:greedy}
\STATE $M= \emptyset$.
\STATE Consider the machines in any order.
\FOR {$i=1$ to $m$}
\STATE Let $k$ be the maximal value such that there are $k$ unmatched jobs with $b(j,i)\geq k$
\STATE Let $A$ be the set of $k$ unmatched jobs with highest $b(j,i)$ value.
\STATE For every $j \in A$, add $(j,i)$ to $M$.
\ENDFOR
\end{algorithmic}
\end{algorithm}

\negB\begin{theorem}
    \label{thm:2-approx}
    Algorithm \ref{alg:greedy} provides a $\frac 1 2$-approximation to the maximum PD-matching.
\end{theorem}
\begin{proof}
    We prove that the algorithm returns a strongly-maximal PD-matching. By Claim~\ref{cl:2approx}, any such matching is a $\frac{1}{2}$-approximation.
Let $j$ be an unmatched job. 
For every machine $i \in V$, when $i$ is considered by the algorithm, $d_i(M)$ is determined to be the maximal number of unmatched jobs that can be matched to $i$, and the jobs with highest tolerance on $i$ are matched to $i$. Since job $j$ is not among these jobs, it must be that $d_i(M) \ge b(j,i)$. Therefore, for any unmatched job $j$ and any machine $i \in V$, it holds that $b(j,i) \le d_i(M)$.     

For every machine, step 4 of the algorithm requires time $O(n \log n)$, thus, the time complexity of the algorithm is $O(m n\log n)$.
\end{proof}

Alternatively, the algorithm could select in every iteration a job $j$ with a global maximal value of $b(j,i)$ such that $d_i(M) < b(j,i)$, and match it to the corresponding machine. 

We show that the analysis is tight {\em even for instances in $\I_{mono}^{Udep}$}. Consider an instance with $n=2k$ jobs, and $2$ machines. For $1 \le j \le k$, let $V_j=\{2\}$ and $b_j=k$.
For $k< j \le 2k$, let $V_j=\{1,2\}$ and $b_j=k$. The set of edges $M=\{(j,1)| k<j\leq 2k\}$ form a strongly-maximal PD-matching, obtained by Algorithm~\ref{alg:greedy}, with $|M|=k$, however, the maximum PD-matching is  $M^*=\{(j,1)| 1 \le j\leq k\}\cup\{(j,2)|k <j \le 2k\}$ of size $|M^*|=2k$. 


Recall that a better, $(1-\frac 1 e)$-approximation, can be achieved using LP rounding techniques~\cite{FGMS11}. We find the analysis of strongly-maximal PD-matching essential, as it is natural, practical and simple.



\section{Hardness Proof for Monotonous Instances}
\label{sec:hardM}
In this section we show that computing a maximum PD-matching is NP-hard even for extremely nicely-structured instances. Recall that an instance is {\em monotonous} if each row and each column in the tolerance matrix is non-decreasing. That is, 
for every $1 \leq j < j' \leq n$ and $1 \leq i < i' \leq m$, it holds that $b(j,i) \leq b(j, i')$ and $b(j,i) \leq b(j',i)$. Practically, monotonicity implies that the $U$-vertices as well as the $V$-vertices can be ordered according to their tolerance and capability. 

Monotonicity is a feature that supposedly makes the problem easier to solve; indicating a greedy assignment of the jobs may result in an optimal assignment. Indeed, simple exchange argument can be used to show that every monotonous instance has a maximum PD-matching in which only the most tolerant jobs and the most capable machines are matched. Formally, 
\begin{observation}
\label{ob:mono}
 Every monotonous instance has a maximum PD-matching $M^*$ such that
for $j^*=n-|M^*|+1$ and some $i^* \in V$, it holds that $(j,i) \in M^*$ if and only if $j \ge j^*$ and $i \ge i^*$.
\end{observation}
\begin{proof}
    Let $M$ be a maximum matching. If $(j_1,i) \in M$, but there is a job $j_2>j_1$ such that $j_2$ is not matched in $M$, consider the matching $M' = (M \setminus \{(j_1, i)\})\cup \{(j_2, i)\}$. From monotonicity, $b(j_2, i) \geq b(j_1, i)$, so $j_2$ is satisfied in $M'$. As the degree of all machines is unchanged, $M'$ is a valid PD-matching, of size $|M|$.

Similarly, consider a machine $i_1$ such that $d_{i_1}(M)>0$ while there is a machine $i_2 > i_1$ such that $d_{i_2}(M) = 0$. Consider a matching $M'= (M\setminus \{(j, i_1)| (j,i_1) \in M\}) \cup \{(j, i_2)| (j,i_1)\in M\}$. From monotonicity, $b(j, i_2) \geq b(j, i_1)$ for every job $j$, so $M'$ is valid, and $|M'|=|M|$.

This process can be repeated until, for $j^* =n-|M|+1$ only the most tolerant $j^*$ jobs are matched to the most capable $i^*$ machines.

\end{proof}

However, as we show, this observation does not lead to an efficient algorithm.

\negB\begin{theorem}
\label{thm:mono_hard}
Finding a maximum PD-matching is NP-hard even for monotonous instances.
\end{theorem}
\begin{proof}
We present a reduction from the $3$-partition problem. The input for $3$-partition is a set of $3k$ integers $A=\{x_1, \dots, x_{3k}\}$, such that $\sum_{\ell=1}^{3k}x_\ell=kB$, and $\frac{B}{4} < x_\ell < \frac{B}{2}$, for all $1 \le \ell \le 3k$. W.l.o.g., we assume the integers in $A$ are sorted in increasing order. The goal is to partition $A$ into $k$ triplets, each with total sum $B$. The $3$-partition problem is known to be strongly NP-hard.

Given an instance of $3$-partition, we construct a monotonous instance of PD-matching, $I=\langle U\cup V, E,\{b(u,v)\} \rangle$ as follows. The set $U$ includes $n=\sum_{\ell=1}^k \ell B=\frac{k(k+1)}{2}B$ vertices (jobs), and the set $V$ includes $3k$ vertices (machines), $i_1, \dots, i_{3k}$.
The jobs are partitioned into $k$ {\em types}, where two jobs $j_1, j_2$ have the same type iff for every machine $i \in V$, $b(j_1,i)=b(j_2,i)$. For every $1 \le \ell \le k$, there are $\ell B$ jobs of type $\ell$. For every job $j$ of type $\ell$, and machine $i$, let $b_(j,i)=x_i \cdot \ell$. Note that $I$ is indeed monotonous.

{\bf Example:} Let $B=100, k=2$, and $A=\{26,30,31,33,36,44\}$ be an input for $3$-partition. The set $A$ has a partition into $A_1=\{26,30,44\}$ and $A_2=\{31,33,36\}$, each sums up to $100$. We construct an input with $3k = 6$ machines and $\sum_{\ell=1}^k \ell B= 100\cdot(1+2)=300$ jobs. For jobs $1 \le j \le 100$ we set $b(j,i)=x_i$, that is, when the tolerances of a job-vertex are represented as a vector of length $m$, $b_j=(26,30,31,33,36,44)$, and for jobs $101 \le j \le 300$ we set $b(j,i)=2x_i$, that is, the tolerance vector is $b_j=(52,60,62,66,72,88)$.

\begin{claim}
The set $A$ has a $3$-partition iff $I$ has a PD-matching of size $n$.
\end{claim}
\begin{proof}
Assume first that $A$ has a partition into triplets $A_1,\dots,A_k$ each having total sum $B$. For each $1 \le \ell \le k$, denote $A_\ell=\{x_{\ell_1}, x_{\ell_2}, x_{\ell_3}\}$. Recall that there are $\ell B$ jobs for which the tolerances of job $j$ are $b_{j}=(\ell\cdot x_1, \dots, \ell \cdot x_{3k})$. We match these jobs to the three machines $i_{\ell_1},i_{\ell_2},i_{\ell_3}$, such that for every $a \in \{1,2,3,\}$, exactly $\ell \cdot x_{\ell_a}$ jobs are matched to machine $a_{\ell_a}$. Since $\ell\cdot x_{\ell_1}+\ell\cdot x_{\ell_2}+\ell\cdot x_{\ell_3} =\ell B$, all jobs of type $\ell$ are matched, and the degree of the machine they are matched to is exactly their tolerance on it.

Back to our example, in a possible matching corresponding to the $3$-partition, the $100$ jobs of $A_1$ are split among machines for which their tolerances are $(26,30,44)$, and the $200$ jobs of $S_2$ split among machines for which their tolerances are $2\cdot(31,33,26)=(62,66,72)$. All $300$ jobs are satisfied. As shown in Figure~\ref{fig:3par}, the PD-matching is not unique, and the crucial point is that for all $\ell$, the $\ell B$ jobs of type $\ell$ are matched to a set of machines corresponding to a triplet.
    \begin{figure}[ht]
\centering
\includegraphics[width=0.9\textwidth]{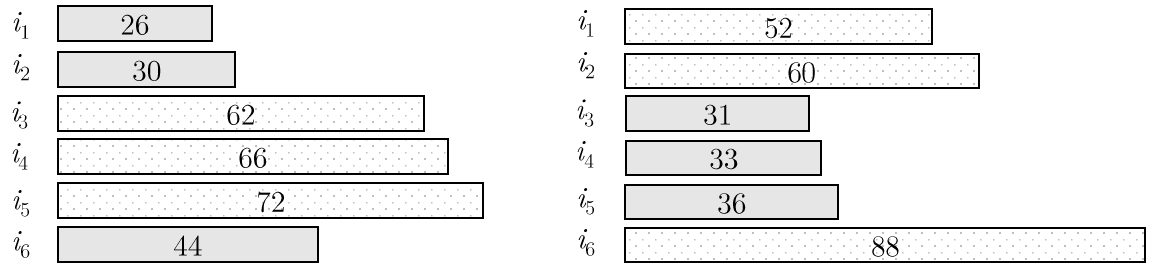} \\
\caption{The only two PD-matchings of size $300$ in the reduction for $A=\{26,30,31,33,36,44\}$. The figure shows the corresponding job assignment. The labels correspond to the degrees of $V$-vertices in the matching, that is, the degree of the machines. Jobs of the same type are matched to machines corresponding to a triplet. Each of the two PD-matchings induces a $3$-partition.}
\label{fig:3par}
\end{figure}

For the other direction of the proof, recall that for every $x \in A$, we have that $\frac{B}{4} < x < \frac{B}{2}$, Therefore, for every $1 \leq \ell \leq k$, jobs of type $\ell$ have a tolerance of less than $\frac{\ell B}{2}$ on each machine. Since there are $\ell B$ jobs of type $\ell$, not all jobs of type $\ell$ can be matched to less than $3$ machines. Since the number of jobs that can be matched to a machine depends on the tolerance of the job with the lowest tolerance that is matched to that machine, and since the integers in $A$ sum up to $kB$, if jobs of any type are matched to more than $3$ machines, there must be some unmatched jobs. Thus, in a PD-matching of size $n$, for each type $\ell$, jobs of each type $\ell$ are matched to exactly $3$ machines, whose total degree is $\ell B$. As each machine corresponds to an element of $A$, a PD-matching of size $n$ induces a $3$-partition.
\end{proof}

Note that $|U|$ is polynomial in $k$ and $B$, and since $3$-partition is NP-hard even if the input is given in unary, our reduction is polynomial. 
\end{proof}
\section{The Classes \texorpdfstring {$\I^{Udep}$ and $\I^{Vdep}$}{of U-dependent and V-dependent Tolerances}}
\label{sec:dep}
In this section we consider instances with $U$-dependent or $V$-dependent tolerances. Recall that $I \in \I^{Udep}$ iff every $j \in U$ is associated with a single tolerance value, $b_j$, and a set $V_j$ of machines such that $b(j,i) \in \{0, b_j\}$ where $b(j,i) =b_j$ iff $i \in V_j$. Similarly, $I \in \I^{Vdep}$ iff every $i \in V$ is associated with a single tolerance value, $b_i$, and a set $U_i$ of jobs such that $b(j,i)\in\{0,b_v\}$, where $b(j,i) =b_i$ iff $j \in U_i$.

It is easy to see that the class $\I^{Vdep}$ is a special case of $b$-matching, where $U$-vertices have bound $1$, $V$-vertices have bound $b_v$, and there is an edge $(j,i)$ if and only if $b(j,i)=b_j$.
Since $b$-matching can be solved efficiently~\cite{Ans87}, we have: 
\begin{theorem}
\label{thm:Vdep}
Computing a maximum PD-matching of $I \in \I^{Vdep}$ can be done efficiently.
\end{theorem}

For the class $\I^{Udep}$, we present three different results.
Our first result, which is simple and positive, refers to instances of identical machines, in which every job can be assigned to all machines. Formally, for every $u \in U$ and $v \in V$, $b(u,v)=b_u$.
\begin{theorem}
\label{thm:UdepComplete}
If for all $u \in U$, it holds that $V_u=V$, then it is possible to compute a maximum PD-matching of $I \in \I^{Udep}$ in time $O(n \log n)$.
\end{theorem}
\begin{proof}
    Algorithm~\ref{alg:greedy_uniform} below produces a maximum PD-matching.

\begin{algorithm}[ht]
\caption{Optimal Algorithm for $\I^{Udep}$ where and $\forall j \in U, V_j=V$}
\begin{algorithmic}[1]
\label{alg:greedy_uniform}
\STATE Sort the vertices in $U$ such that $b_1 \geq \dots \geq b_n$.
\FOR{$i=1$ to $m$}
\STATE Assign a maximal prefix of unassigned vertices in $U$ to vertex $i$. Formally: 
\STATE Let $k$ be the maximal s.t $b_u \geq k$ for each $u$ among the first $k$ unmatched vertices in $U$.
\STATE Add $(u, i)$ to $M$, for each $u$ among the first $k$ unmatched vertices in $U$.
\ENDFOR
\end{algorithmic}
\end{algorithm}

First, the choice of $k$ in step $3$ guarantees that the algorithm returns a valid PD-matching. 
Note that with $U$-dependent bounds, the machines are identical. We show by induction on $m$, that the PD-matching produced by the algorithm matches a maximal number of jobs to $m$ machines.
The base case refers to $m=1$. Since the algorithm matches the jobs with highest tolerance in the first iteration, a maximal number of jobs is matched to a single machine.

For the induction step, consider any PD-matching $M'$ to $m$ machines. Assume, w.l.o.g., that machine $m$ is the machine for $d_m(M')$ is minimal. Let $j$ be the least-tolerant job matched in $M'$. Recall that the jobs are sorted such that $b_1 \geq \dots \geq b_n$. Therefore, $b_j \le b_{|M'|}$.
The tolerance of $j$ implies that $d_m(M') \le b_j \le b_{|M'|}$.
Denote by $L_{m-1}$ and $L'_{m-1}$ the total number of jobs assigned to the first $m-1$ machines in $M$ and $M'$, respectively. 
By the induction hypothesis, $L_{m-1} \ge L'_{m-1}$.
Let $k=b_{|M'|}-(L_{m-1} - L'_{m-1})$. Consider the $m$-iteration of the algorithm. At least $k$ jobs with tolerance at least $k$ are available, and will be matched to the $m$-th machine in $M$. Therefore, $|M| \ge L_{m-1} +k = b_{|M'|} + L'_{m-1} \ge |M'|$.
\end{proof}

Consider next instances with arbitrary $V_j$ sets. Recall that the special case where for all $j$ we have $b_j=1$ is the classical bipartite maximum matching problem, which is efficiently solvable.
We show that even a slight relaxation of the unit-tolerance assumption makes the problem intractable, and hard to approximate. Specifically, we show that the problem is APX-hard in case of restricted assignment and $U$-dependent job tolerances in $\{1,2\}$. 
\begin{theorem}
\label{thm:restricted12}
Let $I$ be an instance of $\I^{Udep}$ in which
for all $j \in U$, $b_j \in \{1,2\}$. Computing a maximum PD-matching in $I$ is APX-hard.
\end{theorem}
\begin{proof}
    In order to prove APX-hardness, we use an $L$-reduction \cite{PY91}, defined as follows:
\negB\begin{definition}
    \label{def:apx}
    Let $\Pi_1, \Pi_2$ be two optimization problems. We say that $\Pi_1$ $L$-reduces to $\Pi_2$ if there exist polynomial time computable functions $f,g$ and constants $\alpha, \beta >0$ such that, for every instance $I \in \Pi_1$ the following holds:
    \begin{enumerate}
        \item $f(I) \in \Pi_2$ such that $OPT(f(I)) \leq \alpha \cdot OPT(I)$.
        \item Given any solution $\varphi$ to $f(I)$, $g(\varphi)$ is a feasible solution to $I$ such that $|OPT(I) - value(g(\varphi))| \leq \beta \cdot |OPT(f(I)) - value(\varphi)|$.
    \end{enumerate}
\end{definition}

Our $L$-reduction is from $3$-bounded $3$-dimensional matching ($3$DM$3$). Our proof uses an idea from the hardness proof in~\cite{LST90}, for the following decision problem: Given an instance of scheduling on unrelated machines, does it have an instance with makespan at most $2$. The similarity between the two reductions highlights the similarity between PD-matching and the problem of scheduling on unrelated machines, both given by an $n \times m$ matrix.

The input to a $3$D-matching problem is a set of triplets $T \subseteq X \times Y \times Z$, where $|X|=|Y|=|Z|=k$. The number of occurrences of every element of $X \cup Y \cup Z$ in $T$ is at most $3$. The number of triplets is $|T| = m$. The goal is to find a maximal subset $T' \subseteq T$, such that every element in $X\cup Y \cup Z$ appears at most once in $T'$. This problem is known to be APX-hard \cite{Kann91}. Moreover, a stronger hardness result shows the problem has an hardness gap even on instances with a matching of size $k$ \cite{Pet94}, namely instances where there is a solution with $k$ triplets.

Given an instance $T \subseteq X \times Y \times Z$ of $3$DM$3$, we construct an instance $f(T)$ of $\I^{Udep}$ in which for every job $j$, $b_j \in \{1,2\}$.

The number of machines is $m=|T|$.
Machine $i$ corresponds to the triple $T_i$, for $1 \le i \le m$.
For every element in $Y \cup Z$ there is one element-job.
The tolerance of an element job is $b_j=2$, and its allowed set of machines consists of machines corresponding to triplets it belongs to.
For every $1 \le a \le k$ there are $t_a -1$ dummy jobs of type $a$. 
The tolerance of a dummy job of type $j$ is $b_j=1$, and its allowed set of machines consists of machines corresponding to triplets of type $j$.
Note that the total number of dummy jobs is $\sum_a (t_a-1) = m - k$, and all together, there are $2k+m-k=m+k$ jobs.

In order to complete the $L$-reduction, we prove the conditions of Definition \ref{def:apx} are met. As a preliminary, we show that if $T$ has a matching $T'$ of size $k$, then $OPT(f(T))=m+k$.

Let $T'\subseteq T$ be a matching of size $k$. We construct a PD-matching $M_{T'}$ of $f(T)$. For every triplet $T_i=(x_a,y_b,z_c)$ in $T'$, let $(y_b, i), (z_c,i) \in M_{T'}$. Additionally, the $t_a-1$ dummy jobs of type $a$ are matched to the other machines of type $a$, one to each. Since every $X$-job is alone on a machine corresponding to a triplet it belongs to, and $Y$-jobs and $Z$-jobs are paired on machines corresponding to a triplet they belong to, the tolerance constraints are all satisfied. 
As all jobs are matched, $OPT(f(T))=|M_{T'}| = m-k + 2k = m+k$.

Given a feasible PD-matching $M$, let $g(M)$ be the set of triplets corresponding to machines on which two jobs are assigned. The theorem follows from the following claims, showing that Definition \ref{def:apx} is satisfied for $\alpha=4$ and $\beta=1$.

\begin{claim}
\label{clm:apx_cond1}
If $OPT(T)=k$ then $OPT(f(T)) \leq 4k$.
\end{claim}
\begin{proof}
    Since every element of $X \cup Y \cup Z$ appears at most $3$ times in $T$, $m \leq 3k$. Thus, $OPT(f(T)) = m+k \leq 3k +k = 4k$.
\end{proof}

\begin{claim}
\label{clm:apx_cond2}
$OPT(T) - |g(M)| \leq (OPT(f(T)) - |M|)$.
\end{claim}
\begin{proof}
     Since we assume that $T$ has a perfect matching, we have that $OPT(T)=k$. Since we proved that $OPT(f(T)) = m+k$, it is sufficient to prove that $k-|g(M)| \leq m+k - |M|$. 
Consider a PD-matching $M$ of $f(T)$. Recall that every machine has $0, 1$ or $2$ jobs matched to it. Thus, $|M| \leq  2|g(M)| + (m-|g(M)|) = m+|g(M)|$. Therefore, $k -|g(M)| = m+k -(m + |g(M)|) \leq m+k - |M|$ as needed.
\end{proof}
\end{proof}

\negC
\subsection{The Class \texorpdfstring{$\I_{mono}^{Udep}$}{of Monotonous U-dependent Tolerances}}
For instances in $\I^{Udep}$, the definition of monotonicity reduces to the following: the $U$-vertices are ordered such that for every $j \le j'$ it holds that $b_j \le b_{j'}$ and $V_j \subseteq V_{j'}$, and the $V$-vertices are ordered such that for every $i \le i'$ it holds that $U_i \subseteq U_{i'}$.
Let $\I_{mono}^{Udep}$ denote the class of monotonous instances in $\I^{Udep}$. 
\begin{theorem}
    A maximum PD-matching of $I \in \I_{mono}^{Udep}$ can be found in time $O(n)$.
\end{theorem}
\begin{proof}
Algorithm~\ref{alg:greedy_mono_Urel} below returns a set of edges $M$.
The algorithm considers the jobs in decreasing order of tolerance. Every job is matched with the least capable machine on which it fits. If a machine is not feasible for a job, it is not considered by later jobs. 

\begin{algorithm}[ht]
\caption{Algorithm for maximum PD-matching in an instance in  $\I_{mono}^{Udep}$}
\begin{algorithmic}[1]
\label{alg:greedy_mono_Urel}
\STATE Sort the $U$-vertices according to tolerance, that is, $U_1\subseteq\cdot\cdot\cdot \subseteq U_m$ and $b_1\le \cdot\cdot\cdot \le b_n$.
\STATE Sort the $V$-vertices according to capability, that is $V_1 \subseteq\cdot\cdot\cdot \subseteq V_n$.
\STATE $j=n$ (most tolerant unassigned job).
\STATE $i=1$ (least capable machine).
\STATE $M = \emptyset$.
\REPEAT
\IF{$i \in V_j$ and $d_i(M)< b_j$} 
\STATE Add $(j,i)$ to $M$.
\STATE $j=j-1$.
\ELSE
\STATE $i=i+1$.
\ENDIF
\UNTIL{$j=0$ (all job are matched) or $i>m$ (no available machines).}
\STATE Return $M$.
\end{algorithmic}
\end{algorithm}

We prove that the algorithm returns a maximum PD-matching. First, we show that the matching is feasible. An edge $(j,i)$ is added to $M$ only if $i \in V_j$ and, before its addition, there are less than $b_j$ jobs matched to $i$. Thus, the tolerance-constraint of $j$ is met when it is assigned. Additional jobs may be matched to $i$ in later iterations. However, since the jobs are considered in non-increasing order of tolerance, each of the added jobs have tolerance at most $b_j$, guaranteeing that the tolerance-constraint of job $j$ is preserved even if additional jobs are matched to $i$. 

Next, we show that $M$ has maximum cardinality. The proof of the following claim is based on a non-standard exchange argument of possibly multiple jobs.
\begin{claim}
    \label{cl:MoptMono}
     $M$ is a maximum PD-matching.
\end{claim}
\begin{proof}
    Let $M^*$ be a maximum PD-matching. We prove that $|M^*|=|M|$. Specifically, we show that $M^*$ can be converted to $M$ by switching the assignment of some jobs, without decreasing the size of the matching or making it invalid.

As long as $M \neq M^*$, let $j$ be the least tolerant job whose assignment in $M$ and $M^*$ is different. Formally, each of the jobs $j+1,\ldots,n$ is matched to the same machine in $M^*$ and in $M$, and for some $i \neq i'$, we have that $(j,i) \in M$, but $(j,i') \in M^*$. Since $M$ and $M^*$ agree on the assignment of jobs $j+1,\ldots,n$, and the algorithm fills the machines in increasing order of capability, it must be that $i' > i$, so $i'$ is more capable than $i$. If $M^* \setminus \{(j,i')\} \cup \{(j,i)\}$ is a valid matching, we are done. 

Otherwise, consider the matching $\hat{M^*}$ that is constructed from $M^*$ as follows. All jobs matched to machines in $V\setminus \{i,i'\}$ remain unchanged. Let $W$ be the set of jobs matched to $i$ or $i'$ in $M^*$. In $\hat{M^*}$, the jobs of $W$ are reassigned to these two vertices: Let $k$ be the maximal number such that there are $k$ jobs in $W$ with $b_j \geq k$ and $i \in V_j$. The $k$ most tolerant jobs in $W$ are matched to $i$, and the remaining $|W|-k$ jobs are matched to $i'$. Note that $\hat{M^*}$  agrees with $M$ on $j, \ldots, n$.  

We prove that $\hat{M^*}$ is feasible. Clearly, the matching is valid for all jobs matched to machines in $V \setminus \{i,i'\}$. Furthermore, by the choice of $k$, it is valid for all jobs matched to $i$.  

If $k=|W|$, we are done. Otherwise, we show that $\hat{M^*}$ is valid for all jobs matched to $i'$: 
Since $k \ge d_i(M^*)$, we have $d_i(\hat{M^*}) \geq d_i(M^*)$. Since $|W|= d_i(M^*)+d_{i'}(M^*) = d_i(\hat{M^*}) + d_{i'}(\hat{M^*})$, the degree of $i'$ in $\hat{M^*}$ fulfills $d_{i'}(\hat{M^*}) \leq d_{i'}(M^*)$. 
Let $j'$ be the least tolerant job in $W$. Since $k < |W|$, job $j'$ is matched to $i'$ in $\hat{M^*}$. If $(j',i') \in M^*$, $j'$ clearly remains valid in $\hat{M^*}$ as $d_{i'} (\hat{M^*}) \leq d_{i'}(M^*)$. Otherwise, $(j',i) \in M^*$. The job-monotonicity implies that $i \in V_j$ for all jobs $j \in W$, thus, $k$ is at least as large as  $d_i(M^*)$. We conclude that $d_{i'}(\hat{M^*}) \leq d_i(M^*)$, and again, $\hat{M^*}$ is valid for $j'$.
Finally, since $j'$ is the least tolerant job in $W$, the job-monotonicity implies that $\hat{M^*}$ is valid for all jobs matched to $i'$ in $\hat{M^*}$.

By repeating the same method as long as $M \neq M^*$, we get an optimal matching that fully agrees with $M$.
\end{proof}

The time complexity analysis of the algorithm is straightforward.
We assume $m \le n$, as otherwise, it is possible to use only the $n$ most capable machines. Every iteration of the repeat loop involves an increase of $i$ or a decrease of $j$. Since the range of $i$ is $1,\ldots,n$ and the range of $j$ is $n,\ldots, 0$ there are at most $n+m$ iterations. Since each iteration takes a constant time, the whole loop takes $O(n+m)$ time.
\end{proof}

Recall that in Theorem~\ref{thm:mono_hard} we proved that maximum PD-matching is NP-hard for arbitrary monotonous instances. A natural question is to analyze the performance of Algorithm~\ref{alg:greedy_mono_Urel} if it is run on monotonous instances without $U$-dependent tolerances. Note that the algorithm is well-defined for any monotonous instance. We show that unfortunately, the approximation ratio of this algorithm is not better than $\frac{1}{2}$.

Consider an instance with $|U|=2k$ and $|V|=k+1$. Each of the first $k$ jobs have tolerance $1$ on all machines. Each of the $k$ last, more tolerant, jobs, has tolerance $1$ on machines $1,\ldots,k$, and tolerance $k$ on the last machine. 
A maximum PD-matching has size $2k$. Specifically, $M^*= \{(j,j)|1 \le j \le k\} \cup \{(j,k+1)|k+1 \le j \le 2k\}$. 
Algorithm~\ref{alg:greedy_mono_Urel} will end up with a matching of size $k+1$ consisting of $\{(2k-j+1,j)|1 \le j \le k+1\}$. Thus, it achieves $\frac{k+1}{2k}$-approximation.
Note that this example captures monotonous instances with tolerance limited to two values, $1$ and $k$.

\negA\section{Restricted Classes of Instances}
\label{sec:restricted}
In this section we analyze the maximum PD-matching problem on restricted classes of instances. For some classes we provide polynomial-time algorithms for computing a maximum PD-matching, while for others we prove it is NP-hard to compute such a matching.
\negC\subsection{Constant Number of Machines}
\label{sec:const_machines}
For a constant number of machines, $m=O(1)$, we present an optimal polynomial-time algorithm. The algorithm is based on guessing the degree of each machine in a maximum PD-matching, and reducing the problem to a bipartite $b$-matching problem.

\begin{theorem}
\label{thm:const_m}
    If $|V|=O(1)$, a maximum PD-matching can be found in polynomial time.
\end{theorem}
\begin{proof}
Let $(d_1, \dots, d_m)$ be a vector such that for all $1 \leq i \leq m$, $0 \leq d_i \leq n$, and $\sum_i d_i \le n$. We present a poly-time algorithm to decide whether there exists a PD-matching $M$ in which for every  $i \in V$, $d_i(M)=d_i$.

Given $(d_1, \dots, d_m)$, consider a $b$-matching problem on a graph $G'=(U \cup V, E')$. The set of edges is $\{(u_j, v_i)~|~ b(j,i) \geq d_i\}$. 
The $b$-matching bound of every $U$-vertex is $1$, and the $b$-matching bound of every vertex $i \in V$ is $d_i$.

It is easy to see that there is a one-to-one correspondence between a $b$-matching $M'$ in $G'$ and a PD-matching of size $|M'|$ in $G$. In particular, if the maximum $b$-matching in $G'$ has cardinality $\sum_i d_i$, then $G$ has a PD-matching in which for every machine $i \in V$ exactly $d_i$ jobs are matched to machine $i$.

The number of candidate vectors is bounded by $n^m$, and since $m = O(1)$, it is possible to find a PD-matching $M$ with a maximal $\sum_{i \in V} d_i$ value.
\end{proof}
\negC
\subsection{Constant Number of Tolerances}
\label{sec:limited}
Let $\mathcal{T} = \{b_{(j,i)}| j \in U, i \in V\}$ be the set of all jobs' tolerances on all machines.
In this section we analyze instances of PD-matching with limited $\mathcal{T}$.

Consider first the case of uniform tolerances, that is, $|\mathcal{T}|=1$. Let $\mathcal{T} = \{k\}$. The corresponding problem is trivially solved by matching $k$ arbitrary jobs to every machine. The cardinality of the resulting PD-matching is $\min(n,km)$, which is clearly optimal.

While this extreme case is trivially solvable, as we show, the problem becomes intractable already for $|\mathcal{T}|=2$. 

\negB\begin{theorem}
\label{thm:card2_hard}
    Let $\mathcal{T} = \{k_1,k_2\}$ where $k_1 \geq 1$ and $k_2 > \max(2, k_1)$. It is APX-hard to find a maximum PD-matching.
\end{theorem}
\begin{proof}
     Let $\mathcal{T} = \{k_1, k_2\}$ where $k_1 \geq 1$ and $k_2 > \max(2,k_1)$. We prove it is APX-hard to find a maximum PD-matching by an $L$-reduction from $d$-bounded $d$-dimensional matching ($d$DM-$d$).

    The input to $d$DM-$d$ is a set of tuples $T \subseteq X_1 \times \dots \times X_{d}$, where $|X_i|=k$ for every $1 \leq i \leq d$. The number of occurrences of every element of $X_1 \cup \dots \cup X_{d}$ in $T$ is at most $d$. The number of tuples is $|T| = t \ge k$. The goal is to find a subset $T' \subseteq T$, such that every element in $X_1 \cup \dots \cup X_{d}$ appears exactly once in $T'$. This problem is known to have an NP-hard gap even for instances with a solution of $k$ triplets \cite{Pet94}.

    To prove hardness for $\mathcal{T} = \{k_1,k_2\}$ where $k_1 \geq 1$ and $k_2 > \max(2,k_1)$, we consider an instance $T$ of $k_2$-bounded $k_2$-dimensional matching. Note that $T \subseteq X_1 \times \dots \times X_{k_2}$, where $|X_i|=k$ for every $1 \leq i \leq k_2$. We construct an instance of PD-matching $S$ with $k \cdot k_2 + k_1(t-k)$ jobs and $t$ machines. Every machine corresponds to a tuple in $T$. For every $w \in X_1 \cup \dots \cup X_{k_2}$ there is a corresponding job $j_w$. For every $w \in X_1 \cup \dots \cup X_{k_2}$ and every $\ell \in T$, for the corresponding job $j_w$ and machine $a_\ell$, $b(j_w, a_\ell)= k_2$ if $w \in \ell$ and $b(j_w,a_\ell)= k_1$ otherwise. There are also an additional $k_1(t-k)$ dummy jobs which have a tolerance of $k_1$ for every machine.

    \begin{claimstar}
        There is a PD-matching $M$ of size $|U|$ if and only if there is a subset $T' \subseteq T$ with $|T'|=k$ where every element in $X_1 \cup \dots \cup X_{k_2}$ occurs once.
    \end{claimstar}
    \begin{proof}
    Assume there is a matching $M$ of size $|U|$. Since there are $t$ machines, $k_1(t-k)$ dummy jobs with tolerance $k_1$ for every machine, and $k \cdot k_2$ non-dummy jobs, there are exactly $k$ machines without a dummy job matched to them in $M$. For each of those machines $a_\ell$ where $\ell \in T$ there are exactly $k_2$ non-dummy jobs with tolerance $k_2$ matched to $a_\ell$ - which must be the jobs corresponding to the $k_2$ elements of $\ell$. Thus, as there are exactly $k \cdot k_2$ jobs corresponding to elements of $X_1 \cup \dots \cup X_{k_2}$, all jobs are matched if and only if the $k$ tuples corresponding to the machines that the dummy jobs are not matched to are a $k_2$-matching.
    
    In the other direction, given $T' \subseteq T$ with $|T'|=k$ where every element in $X_1 \cup \dots \cup X_{k_2}$ occurs once, we can construct a PD-matching of size $|U|$. For every tuple $\ell \in T'$, all $k_2$ jobs corresponding to the elements of $\ell$ are matched to the machine $a_\ell$. The remaining $k_1(t-k)$ dummy jobs are matched in a balanced way to the remaining $t-k$ machines. Thus, all jobs are matched.
    \end{proof}

    Similarly to Theorem \ref{thm:restricted12}, as we know that all jobs can be matched in the constructed instance, we can prove the conditions of Definition \ref{def:apx} are satisfied for $\alpha=dk_2, \beta=1$.
\end{proof}

On the other hand, for the remaining cases of $|\mathcal{T}|=2$, we show a polynomial-time algorithm for computing a maximum PD-matching.
\negB\begin{theorem}
\label{thm:card2_alg}
Let $\mathcal{T}=\{0,k\}$ or $\mathcal{T} = \{1,2\}$. A maximum PD-matching can be found in polynomial time.
\end{theorem}
\negB\begin{proof}
If $\mathcal{T}=\{0,k\}$, the instance is a special case of the class $\I^{Vdep}$, for which, as shown in Theorem~\ref{thm:Vdep}, an optimal efficient algorithm exists.

For the case $\mathcal{T}=\{1,2\}$, we describe an algorithm, adapted from~\cite{Sch83}, that decides, for a given $x$ if there is a PD-matching where exactly $x$ machines have $2$ jobs matched to them. 

Given $x$, construct a graph $G_x=(V_x,E_x)$. $V_x = U \cup \{v_i^1, v_i^2~|~i \in V\} \cup \{w_i~|~ 1 \leq \ell \leq n-2x\}$ where we denote vertices $u \in U$ as job-vertices, $v$-vertices as machine-vertices, and $w$-vertices as dummy-vertices. $E_x = \{(u,v^1_i), (u,v^2_i)~|~u \in U, i \in V, b(j,i) = 2\} \cup \{(v_i^1, v_i^2)~|~ i \in V\} \cup \{(u,w_\ell)~|~u \in U, 1 \leq \ell \leq n-2x\}$. 

Recall that a perfect matching in $G_x$ is a subset of the edges, such that every vertex is in exactly one edge. Assume there is a perfect matching $E'_x \subseteq E_x$ in $G_x$. There are $n-2x$ dummy-vertices which must be matched to job-vertices, so a perfect matching corresponds to a valid PD-matching in $G$ of size $2x$, which are the jobs corresponding to job-vertices matched to machine-vertices in $G_x$. Additionally, since in a perfect matching all vertices are matched, for every machine $i$, its corresponding machine-vertices $v_i^1, v_i^2$ must either be matched to some vertices $u_1, u_2 \in U$ where $b(u_1, i)=b(u_2,i)=2$, or be matched to each other. Thus, in the corresponding PD-matching in $G$, all machines have either $2$ or $0$ jobs matched to them.

In the other direction, given a PD-matching $M$ in $G$ with $2x$ satisfied jobs, such that every machine has either $0$ or $2$ jobs matched to it,  the corresponding graph $G_x$ has a perfect matching. Formally, for every two jobs $u_1, u_2 \in U$ and machine $i$ such that $(u_1, i), (u_2,i) \in M$, $(u_1, v_i^1), (u_2,v_i^2) \in E'_x$. For every machine $i$ where $d_i(M)=0$, $(v_i^1, v_i^2) \in E'$, and for every unmatched job $j$ $(v_j, w_\ell) \in E'_x$ for some distinct $1 \leq \ell \leq n-2x$.

For every $x$, if there is a perfect matching in $G_x$, in the corresponding PD-matching $M$ in $G$, all $2x$ matched jobs have tolerance $2$ for the machine they are matched with. Since all remaining jobs have tolerance at least $1$ on all machines, if every unmatched job is matched to some machine with degree $0$, there are $min(n, 2x + (mx) = min(n,m+x)$ matched jobs in the resulting PD-matching. Binary search can be applied to find the maximal value of $x$ for which $G_x$ has a perfect matching. By the above discussion, the corresponding PD-matching is maximal.
\end{proof}

We are left to consider instances with $|\mathcal{T}| \geq 3$. 

\negB\begin{theorem}
     If $|\mathcal{T}|\geq 3$, it is APX-hard to find a maximum PD-matching.
\end{theorem}
\begin{proof}
    If $|\mathcal{T}| \geq 3$ and $\mathcal{T} \neq \{0,1,2\}$, then there must be $k_1, k_2 \in \mathcal{T}$ such that $k_1 > 1$ and $k_2 > \max(2,k_1)$. Therefore, Theorem \ref{thm:card2_hard} covers this case. 

    For $\mathcal{T} = \{0,1,2\}$, recall that in Theorem~\ref{thm:restricted12}, we proved APX-hardness for $U$-dependent tolerance, and $b_j \in \{1,2\}$.
    Clearly, allowing $b(j,i)=0$ is equivalent to restricting the assignment of $j$ to some $V_j \subseteq V$. Therefore, Theorem~\ref{thm:restricted12} covers a special case of $\mathcal{T} = \{0,1,2\}$.
\end{proof}

\negB\noindent{\bf Monotonous Instances with  $|\mathcal{T}|=3$:}

Recall that in a monotonous instance there is an order of the jobs and an order of the machines, such that for every jobs $j_1 \leq j_2$ and machines $i_1 \leq i_2$, $b(j_1, i_1) \leq b(j_1, i_2)$ and $b(j_1, i_1) \leq b(j_2, i_1)$. In this section we provide an optimal algorithm for monotonous instances with $|\mathcal{T}|=3$. 
Let $\mathcal{T}= \{k_1, k_2,k_3\}$ where $0\leq k_1 < k_2<k_3$. If $n \leq mk_1$ clearly all jobs can be matched, therefore, in the sequel we assume $n> mk_1$.

We begin with several useful claims on maximum PD-matchings in instances. The proofs are based on exchange arguments. 

\begin{claim}
    \label{clm:mono_limited_special}
    Consider a monotonous instance of PD-matching with $\mathcal{T} = \{k_1, k_2, \dots, k_{|\mathcal{T}|}\}$ such that $|\mathcal{T}| \geq 2$ and $0 \leq k_1 < k_2 < \dots < k_{|\mathcal{T}|}$. There is a maximum PD-matching $M'$ such that there are at most $|\mathcal{T}|-1$ machines $i$ with $d_i(M) \notin \mathcal{T}$.
\end{claim}
\begin{proof}
     Let $M$ be a maximum PD-matching with at least $|\mathcal{T}|$ machines $i$ such that $d_i(M) \notin \mathcal{T}$. We can assume that no machine $i'$ has a degree less than $k_1$, as every job has a tolerance of at least $k_1$ on every machine, so jobs from other machines can be moved to $i'$. Thus, there is some $1<\ell\leq |\mathcal{T}|$ such that there are two machines, $i_1 \leq i_2$, for which both $k_{\ell-1} < d_{i_1}(M) < k_{\ell}$ and $k_{\ell-1} < d_{i_2}(M) < k_{\ell}$. Note that as $\mathcal{T}$ is the set of tolerances, and from monotonicity, for every job $j$ matched to $i_1$ or $i_2$, $k_\ell \leq b(j,i_1) \leq b(j, i_2)$.

    Thus, we construct a matching $M'$, by replacing matchings $(j, i_2)$ with $(j, i_1)$, until $d_{i_1}(M') = k_\ell$, or $d_{i_1}(M') = k_{\ell-1}$. Note that all jobs remain matched in $M'$. Thus, by repeating this process, a matching meeting the restrictions is reached.
\end{proof}

Moreover, we show that only the worst machines are matched to the $k_1$ worst jobs.

\begin{claim}
    \label{clm:mono_limited_worst}
    Consider a monotonous instance of PD-matching with $\mathcal{T} = \{k_1, k_2, \dots, k_{|\mathcal{T}|}\}$ such that $|\mathcal{T}| \geq 2$ and $0 \leq k_1 < k_2 < \dots < k_{|\mathcal{T}|}$. There is a maximum PD-matching $M'$ for which there is a machine $i'$ and a job $j'$, such that if $d_i(M') \leq k_1$, then $i \leq i'$, and the only jobs matched to such machines are jobs $j$ for which $j \leq j'$.
\end{claim}

\begin{proof}
    Recall that the tolerance of every job is at least $k_1$. Thus, clearly there is a maximum PD-matching where every machine has a degree of at least $k_1$.

Consider a a maximum PD-matching $M'$. Due to the monotonicity, using an exchange argument we can assume the lowest degree machines are the least tolerated machines. Similarly, we can assume the jobs matched to them are the lowest tolerance jobs.
\end{proof}

The exchange arguments used for Observation~\ref{ob:mono}, Claim~\ref{clm:mono_limited_special}, and Claim~\ref{clm:mono_limited_worst} can be used sequentially on the same matching. Thus we have the following:

\begin{corollary}
    \label{cor:mono_organized}
    For every monotonous instance with tolerances $\mathcal{T} = \{k_1, \dots, k_{|\mathcal{T}|}\}$ such that $0 \leq k_1 < \dots < k_{|\mathcal{T}|}$, there is a maximum PD-matching $M'$ such that only the $j'$ most tolerant jobs are matched at a degree higher than $k_1$, they are matched to only the $i'$ most capable machines, all remaining machines are matched to $k_1$ jobs, and for every $\ell > 1$, there is at most one machine $i$ with $k_{\ell-1}<d_i(M)<k_{\ell}$.
\end{corollary}

We can now use this corollary, to describe and prove an algorithm for monotonous instance of PD-matching with limited tolerances.

Our algorithm consists of two steps. In the first step, we guess the matching size, and the machines that have a degree higher than $k_1$. We also guess at most $|\mathcal{T}|-1$ machines that have a degree that is not in $\mathcal{T}$, and their degree. In the second step, we greedily assign the jobs, starting from the worst job that needs to be matched, to the worst machine for which it has a maximal tolerance.

\begin{algorithm}[ht]
\caption{Algorithm for monotonous instances with tolerances $\mathcal{T} = \{k_1, k_2, k_3\}$}
\begin{algorithmic}[1]
\label{alg:mono_limited}
\STATE Guess $0 \leq n' \leq n$, and $0 \leq m' \leq m$.
            \STATE Let $\mathcal{U}$ be the set of $n'$ most tolerant jobs, and $\mathcal{V}$ be the set of $m'$ most capable machines.
            \STATE Guess a set $\mathcal{S}$ of at most $2$ machines, and a target degree $tar(i)$ for each $i \in \mathcal{S}$.
            \STATE Match the least tolerant jobs in $\mathcal{U}$ to each machine in $V \setminus \mathcal{V}$, until each has a degree of $k_1$.
            \STATE Assign the $tar(i)$ least tolerant unmatched jobs with a tolerance of at least $tar(i)$ to each machine $i \in \mathcal{S}$.
            \WHILE{Some job in $\mathcal{U}$ is unmatched, and some machine $i \in \mathcal{V}$ has $d_i = 0$}
                \STATE Let $j'$ be the least tolerant unmatched job in $\mathcal{U}$
                \STATE Choose a maximal $k \in \mathcal{T}$ such that there is some machine $i \in \mathcal{V}$ with no matched jobs, with $b(j', i) \geq k$.
                \STATE Let $i$ be the least-tolerated machine for which $b(j', i)\geq k$.
                \STATE Add $(j,i)$ to $M$ for the $k$ least tolerant unmatched jobs $j$.
            \ENDWHILE
\STATE Return the matching with a maximum number of matched jobs.
\end{algorithmic}
\end{algorithm}

\negB\begin{theorem}
\label{thm:monoT3}
    Algorithm \ref{alg:mono_limited} computes a maximum PD-matching for monotonous instances with tolerances $\mathcal{T}=\{k_1, k_2, k_3\}$, in polynomial time.
\end{theorem}
\begin{proof}
    We prove that the algorithm returns a maximum PD-matching. First, from Corollary \ref{cor:mono_organized}, there is a maximum PD-matching $M$ such that there exists $n' \leq n$ such that only the $n'$ most tolerant jobs are matched. Additionally, every machine $i \notin \mathcal{V}$ has $k_1$ of the least tolerant jobs matched to it, and all remaining jobs in $\mathcal{U}$ are only matched to machines in $\mathcal{V}$. For each of these machines, it has either $k_2$ or $k_3$ jobs matched to it, except at most $2$ machines. Thus, at least one guess of $n', m', \mathcal{S}$ corresponds to an optimal matching.

    For the corresponding guess, consider such an optimal PD-matching $OPT$ constructed by matching the jobs in increasing tolerance, and consider the PD-matching $M$ constructed by Algorithm \ref{alg:mono_limited}. We prove that $M$ is an optimal matching. Note that by the choice of $OPT$, all machines $i \in V \setminus \mathcal{V}$ have $d_i(OPT) = k_1$, and all other machines $i \in \mathcal{V}$ have $d_i(OPT) \in \{k_2, k_3\}$, except at most $2$ machines. Note that by definition, the jobs matched to machines in $(V \setminus \mathcal{V}) \cup \mathcal{S}$ are the least tolerant jobs, so from monotonicity we can assume the same jobs are matched to the same machines in $OPT$.

    If $M \neq OPT$, let $j$ be the least tolerant job that is matched to different machines in $M$ and $OPT$. Let $i_M, i_{OPT}$ be the machines it is matched to in $M, OPT$ respectively. Recall that we assume $OPT$ is constructed by matching the jobs in increasing order of tolerance, and note that as the machines are considered one at a time in the algorithm, $j$ is the first job that is matched to $i_{OPT}$ in the construction of $OPT$. We divide into cases.
    
    If $d_{i_{OPT}}(OPT) = k_2$ and $d_{i_M}(M)=k_3$, or if $d_{i_M}(OPT) = d_{i_{OPT}}(OPT)$, as the jobs are matched in order of increasing tolerance, $j$ can be switched with the next job matched to $i_M$ in $OPT$, resulting in a maximum PD-matching $M'$ with $j$ matched to $i_M$.

    We are left with the case $d_{i_M}(OPT) = k_3$ and $d_{i_{OPT}}(OPT)=k_2$. If $b(j,i_{M}) = k_3$, implying $d_{i_M}(M)= k_3$, since the jobs are matched in increasing tolerance, we can exchange $j$ with the next job matched to $i_M$ in $OPT$. Otherwise $b(j,i_{M}) = d_{i_M}(M)=k_2$, and by definition of Algorithm \ref{alg:mono_limited}, $i_M$ is the least capable machine which $j$ has a tolerance of $k_2$ for, so $i_{OPT}$ is less capable than $i_M$, and the degree of $i_M$ when $j$ is matched is $0$. Thus, both $i_M$ and $i_{OPT}$ have degree $0$ before $j$ is matched, so the PD-matching $OPT'$, in which $i_{OPT}$ and $i_M$ are switched, is a maximum PD-matching. Since $(j, i_M) \in OPT'$ and all jobs less tolerant than $j$ are matched to the same machine as in $OPT$, $M$ is a maximum PD-matching.

    Time complexity analysis: By Observation~\ref{ob:mono}, there are only $n$ guesses of $\mathcal{V}$, and $m$ guesses of $\mathcal{V}$. From Claim \ref{clm:mono_limited_special} there are only $\binom{m}{2}n^2$ guesses for $\mathcal{S}$. The greedy phase of the algorithm takes polynomial time. Thus, the whole algorithm is polynomial.
\end{proof}

\subsection{Constant Number of Job Types}
\label{sec:limited_types}
The next class for which we show that the problem is tractable consists of instances with a constant number of job types. Formally, we say that two jobs $j_1, j_2$ are of the same {\em type} $\tau$ if for every machine $i$, $b(j_1, i)=b(j_2,i)$. We denote by {\em type} the corresponding set of jobs. A type is characterized by a tolerance vector $\tau$, where $b^{\tau}_i = b(j,i)$ for all the jobs of type $\tau$.

Denote by $t(I)$ the number of job-types in an instance $I$ of PD-matching.
If $t(I)=1$ then we have $V$-dependent tolerances, for which an optimal algorithm is given in Section~\ref{sec:dep}.
For $t(I)=2$, we present an optimal algorithm whose time complexity is $O(mn)$, and for $t(I) \ge 3$ we present an algorithm whose time complexity is $O(t(I)(nm)^{poly(t(I))})$, which is polynomial for a constant number of types. Our algorithms are based on reducing the solution space to be considered as follows.

\negB\begin{claim}
\label{clm:unique_shared}
    For any PD-matching $M$, and any two types $\tau_1, \tau_2$, there exists a PD-matching $M'$ of size $|M|$, such that at most one machine in $M'$ is assigned jobs of both types.
\end{claim}
\begin{proof}
    Let $M$ be a PD-matching in which two machines, $i_1$ and $i_2$, are matched to jobs from the same two types. Let $b^{\tau_1}, b^{\tau_2}$ be the tolerance vectors of the two types. For $\tau' \in \{\tau_1, \tau_2\}$ and $i' \in \{i_1,i_2\}$, let $\ell_M(\tau',i')$ be the number of jobs of type $\tau'$ that are matched to $i'$ in $M$. 

The feasibility of $M$ implies that for $i' \in \{i_1,i_2\}$, $d_{i'}(M) \leq min(b^{\tau_1}(i'), b^{\tau_2}(i'))$. Assume, w.l.o.g., that $\ell_M(\tau_1,i_1)$ is minimal among the four $\ell_M(\tau',i')$ values. We modify $M$ by swapping the assignment of the $\ell_M(\tau_1,i_1)$ jobs of type $\tau_1$ that are matched to $i_1$ and same number of jobs of type $\tau_2$ that are matched to $i_2$ in $M$. Clearly, the matching size does not change, and the tolerance constraints are fulfilled. In the resulting PD-matching, no job of type $\tau_1$ is matched to $i_1$. The process can be repeated as long as some pair of machines are both assigned jobs from the same two types.
\end{proof}

Intuitively, by the above claim, if $t(I)$ is a constant, then it is sufficient to allow only a constant number of machines to be assigned jobs of multiple types. Each of the remaining machines is matched only with jobs of a single type.

Our algorithm for $t$ types iterates over all options for selecting machines that process more than one job-type, and all options for the number of jobs of each type matched to them. For the remaining machines, we show that the problem is similar to a multidimensional knapsack problem with a constant number of dimensions, and we use dynamic programming to decide the job-type matched to each machine. 

\negB\begin{theorem}
\label{thm:maxsat-ttype}
Maximum PD-matching with $t$ job-types is solvable in time $O(t(mn^t)^{t^2-t+1})$.
\end{theorem}
\begin{proof}
    We show that Algorithm \ref{alg:maxsat-ttype} solves the problem optimally.

\begin{algorithm}[ht]
\caption{Finding a maximum PD-matching with $t$ types of jobs $\tau_1, \dots, \tau_t$}
\begin{algorithmic}[1]
\label{alg:maxsat-ttype}
\STATE{Guess a set $\mathcal{S}$ of at most $t(t-1)$ machines.\label{alg:maxsat-ttype:machine-loop}}
\STATE{Guess a valid matching $M_{\mathcal{S}}$ of job types to $\mathcal{S}$.}
\STATE Denote the number of unmatched jobs of each type by $q_1, \dots, q_t$. W.l.o.g., denote the $m'$ machines not chosen in step \ref{alg:maxsat-ttype:machine-loop} by $v_1, \dots, v_{m'}$. 
\STATE Solve the following knapsack problem variant using dynamic programming:
\vspace{-0.5\baselineskip}
\begin{align*}
        (5.1) \quad & \forall i \quad K[i;0,\dots 0] = 0\\
        (5.2) \quad & \forall w_1 \leq q_1,\dots, \forall w_t\leq q_t \quad K[0;w_1, \dots, w_t] = 0\\ 
        (5.3) \quad &\forall i>0,\forall w_1 \leq q_1,\dots, \forall w_t\leq q_t\quad K[i; w_1, \dots, w_t] = \\
        & max (K[i-1; w_1, \dots, w_t], \\
        & K[i-1; w_1 - min(w_1,\tau_1(v_i)), w_2, \dots w_t] + min(w_1,\tau_1(v_i)),\\
        &\dots, \\
        & K[i-1;w_1, \dots, w_{t-1}, w_t - min(w_t,\tau_t(v_i))] + min(w_t,\tau_t(v_i)))\\
        (5.4) \quad &\text {Calculate $K[m'; q_1, \dots q_t]$}
      \end{align*}
      \vspace{-1.5\baselineskip}
\STATE Denote the matching corresponding to $K[m'; q_1, \dots, q_t]$ by $\hat{M_{\mathcal{S}}}$
\STATE Set $M = M_{\mathcal{S}}\cup \hat{M_{\mathcal{S}}}$. 
\STATE Return the matching $M$ with a maximum number of matched jobs.
\end{algorithmic}
\end{algorithm}

By Claim \ref{clm:unique_shared}, for every two types $\tau, \tau'$ there is at most one machine to which jobs of both types are matched. Therefore, there are at most $(m+1)^{t(t-1)}$ guesses of $\mathcal{S}$, as there are at most $t(t-1)$ {\em modular} machines, and each of them can be any one of the $m$ machines, or not exist at all.

For each choice of the modular machines $\mathcal{S}$, each modular machine may have $O(t)$ different job types matched to it. Since a total of at most $n$ jobs are matched to each of the machines, and the order does not matter, there are $O(n^t)$ ways to match jobs to each of the machines, for $O(n^{t^2(t-1)})$ guesses of a valid matching $M_{\mathcal{S}}$ of jobs types to the machines in $\mathcal{S}$. Note that indeed this combination of guesses for $\mathcal{S}$ and $M_{\mathcal{S}}$ checks every possible choice of modular machines, and every possible matching of jobs to them.

Since there are $n$ jobs in total, the remaining quantity of each job is $O(n)$. By Claim \ref{clm:unique_shared}, in all remaining machines only jobs of at most one type are matched. Note that if jobs of type $\tau$ are matched to a machine $i$, there is a maximum PD-matching where $\tau(i)$ jobs of type $\tau$ are matched to $\tau$ (or all remaining unmatched jobs).

Therefore, the algorithm proceeds to solve the problem as a variant of the multidimensional knapsack problem, with the remaining jobs of each type as the allowed capacity of each knapsack, and the machines as the items, with each item having a different weight and value. This weight and value, which depends on the type of the job, is the tolerance of the jobs of the type for the machine. Formally, $K[i, w_1, \dots, w_t]$ is the size of the maximum matching using $0\leq w_\ell \leq q_\ell$ jobs of type $\ell$ for every $0 \leq \ell \leq t$, and matching only to the first $i$ machines. Each job of type $\ell$ has a weight and value $\tau_\ell(v_i)$ on machine $v_i$. For $i=0$, and for every $w_1 \leq q_1, \dots, w_t \leq q_t$, the value of $K$ is $K[0;w_1,\dots,w_t]=0$. Then, for every $i>0$ we compute $K[i;\dots]$ using $K[i-1;\dots]$ by guessing which type $0\leq \ell \leq t$ maximizes is matched to machine $i$, matching $min(w_\ell, \tau_\ell(v_i))$ jobs of type $\ell$ to $v_i$, which is the maximal possible, We choose $\ell$ to be the type $\ell$ for which the value of $K$ is maximal, after matching the jobs.

Since all item weights are integers, the dynamic programming solution can be found in $O(tmn^t)$. Therefore, in total the running time is $O(tm^{t^2-t+1}n^{t(t^2-t+1)}) = O(t(mn^t)^{t^2-t+1})$ which is polynomial in $n$ and $m$.
\end{proof}

For the case of $t=2$, we present a more efficient algorithm.

\negB\begin{theorem}
\label{thm:maxsat-2type}
Maximum PD-matching with two job-types is solvable in time $O(mn)$.
\end{theorem}
\begin{proof}
Let $V_1 = \{i \in V~|~\tau_1(i) > \tau_2(i)\}$, and its complement $V_2=\{i \in V~|~\tau_2(i) \leq \tau_1(i)\}$. That is, if matched to jobs of only a single type, machines in $V_\ell$ can process more jobs of type $\ell$.

Let $n_1, n_2$ be the number of jobs of type $\tau_1, \tau_2$, respectively. Assume first that for both $\ell=1,2$, it holds that $n_\ell \geq \sum_{i \in V_\ell} \tau_\ell(i)$. In this case, the maximum PD-matching $M$ is clearly a matching such that for every machine $i \in V_\ell$, $d_i(M)=\tau_\ell(i)$. Similarly, if for both $\ell=1,2$,  $n_\ell<\sum_{i \in V_\ell} \tau_\ell(i)$, then a PD-matching in which all jobs are matched is easy to obtain. Therefore, we assume, w.l.o.g., that $n_1 < \sum_{i \in V_1} \tau_1(i)$, and $n_2 \ge \sum_{i \in V_2} \tau_2(i)$.

In our algorithm, a set $H$ of machines in $V_1$ is matched only to jobs of type $\tau_1$. The remaining jobs of type $\tau_1$ are matched to a single machine $i \in V_1 \setminus H$, such that all jobs of type $\tau_1$ are matched. Then, jobs of type $\tau_2$ are matched to machines in $V_2$ and to machine $i$. We use dynamic programming with a table of size $O(mn)$ to select the set $H$.
\end{proof}



\section{Conclusions and Open Problems}
\label{sec:conclusions}
In this paper we extended the vast literature on bipartite matching by studying a setting with pair-dependent bounds. While this setting arises naturally in real-life applications, to the best of our knowledge, it has not been previously considered by the algorithms or scheduling communities. We analyzed the problem of maximizing the matching size with pair-dependent bounds, which corresponds to maximizing the system's throughput in a scheduling environment with time-sharing processing and jobs having machine-dependent tolerances.

The similarity to other matching problems, such as bipartite $b$-matching and scheduling on unrelated machines, and to other assignment problems provides a large selection of known results and techniques to draw inspiration from. However, as we showed, PD-matching varies from these similar models. For example, while maximum bipartite $b$-matching can be solved in polynomial time even in the weighted case, we showed that computing a maximum PD-matching is APX-hard even for instances with $U$-dependent bounds limited to $b_j \in \{1,2\}$. 

Our results demonstrate that a scheduling problem that is solvable in the sequential processing model, may become computationally hard when time-sharing processing is introduced. We also showed that having job-machine related parameters may not make a problem harder if the instance is monotonous, or otherwise limited, but increases the problem complexity in the general case, even if the corresponding $n \times m$  matrix includes only two values. An interesting direction for future work is to investigate the impact of time-sharing processing as well as other job-machine parameters on additional classical scheduling environments.

A natural extension of our problem is the case of weighted PD-matching. As weighted bipartite $b$-matching is solvable in polynomial time, some of our results can be easily extended to the weighted case. For other results, their extension to weighted edges is non-trivial and potentially interesting. 

The $(1-\frac 1 e)$-approximation algorithm due to the reduction to SAP~\cite{FGMS11} is valid also for weighted instances. Coming up with better approximation algorithms, in general or for restricted classes of instances, is an additional direction for future work. Alternatively, it may be possible to strengthen the hardness of approximation result in~\cite{FGMS11} in a way that captures PD matching.



\normalsize

\ifdefined\inmainpaper
  \else
\small
\bibliographystyle{plainurl}

  \end{document}
\fi